\documentclass[10pt,twocolumn,aps,pra,superscriptaddress,showpacs,tightenlines,pdflatex,longbibliography]{revtex4-2}
\usepackage{newpxtext,newpxmath}

\let\coloneqq\relax

\usepackage[utf8]{inputenc}
\usepackage{amsthm}
\usepackage{amssymb}
\usepackage{amsmath}
\usepackage{bbold}
\usepackage{bbm}
\usepackage[pdftex, backref=page]{hyperref}
\usepackage{braket}
\usepackage{dsfont}
\usepackage{mathdots}
\usepackage{mathtools}
\usepackage{enumerate}
\usepackage[shortlabels]{enumitem}
\usepackage{csquotes}
\usepackage{stmaryrd}
\usepackage[cal=boondox]{mathalfa}
\usepackage{graphicx}
\usepackage{stackengine}
\usepackage{scalerel}
\usepackage{tensor}       %\tensor[_n]{\braket{j|\psi}}{_{1\ldots n}}
\usepackage{array}
\usepackage{makecell}
\newcolumntype{x}[1]{>{\centering\arraybackslash}p{#1}}
\usepackage{tikz}
\usepackage{pgfplots}
\usetikzlibrary{shapes.geometric, shapes.misc, positioning, arrows, arrows.meta, decorations.pathreplacing, decorations.pathmorphing, patterns, angles, quotes, calc}
\usepackage{booktabs}
\usepackage{xfrac}
\usepackage{siunitx}
\usepackage{centernot}
\usepackage{comment}
\usepackage{chngcntr}
\usepackage{caption}
\usepackage{subcaption}

\newtheorem{thm}{Theorem}
\newtheorem*{thm*}{Theorem}
\newtheorem{prop}[thm]{Proposition}
\newtheorem*{prop*}{Proposition}

\newtheorem*{lemma*}{Lemma}

\newtheorem*{cor*}{Corollary}

\newtheorem*{cj*}{Conjecture}

\newtheorem*{Def*}{Definition}

\newtheorem*{question*}{Question}

\newtheorem*{problem*}{Problem}

\newtheorem*{assumption*}{Assumption}

\makeatletter
\def\thmhead@plain#1#2#3{%
  \thmname{#1}\thmnumber{\@ifnotempty{#1}{ }\@upn{#2}}%
  \thmnote{ {\the\thm@notefont#3}}}
\let\thmhead\thmhead@plain
\makeatother

\theoremstyle{definition}
\newtheorem{rem}[thm]{Remark}

\newcommand{\bb}{\begin{equation}\begin{aligned}\hspace{0pt}}
\newcommand{\bbb}{\begin{equation*}\begin{aligned}}
\newcommand{\ee}{\end{aligned}\end{equation}}
\newcommand{\eee}{\end{aligned}\end{equation*}}
\newcommand*{\coloneqq}{\mathrel{\vcenter{\baselineskip0.5ex \lineskiplimit0pt \hbox{\scriptsize.}\hbox{\scriptsize.}}} =}

\renewcommand{\epsilon}{\varepsilon}

\newcommand{\id}{\mathds{1}}

\DeclareMathOperator{\Tr}{Tr}

\DeclareMathAlphabet{\pazocal}{OMS}{zplm}{m}{n}

\DeclareMathOperator{\supp}{supp}

\DeclareMathOperator{\diag}{diag}

\newcommand{\lsmatrix}{\left(\begin{smallmatrix}}
\newcommand{\rsmatrix}{\end{smallmatrix}\right)}

\stackMath

\stackMath

\makeatletter
\newcommand*\rel@kern[1]{\kern#1\dimexpr\macc@kerna}
\newcommand*\widebar[1]{%
  \begingroup
  \def\mathaccent##1##2{%
    \rel@kern{0.8}%
    \overline{\rel@kern{-0.8}\macc@nucleus\rel@kern{0.2}}%
    \rel@kern{-0.2}%
  }%
  \macc@depth\@ne
  \let\math@bgroup\@empty \let\math@egroup\macc@set@skewchar
  \mathsurround\z@ \frozen@everymath{\mathgroup\macc@group\relax}%
  \macc@set@skewchar\relax
  \let\mathaccentV\macc@nested@a
  \macc@nested@a\relax111{#1}%
  \endgroup
}

\counterwithin*{equation}{part}
\counterwithin*{thm}{part}
\counterwithin*{figure}{part}

\tikzset{meter/.append style={draw, inner sep=10, rectangle, font=\vphantom{A}, minimum width=30, line width=.8, path picture={\draw[black] ([shift={(.1,.3)}]path picture bounding box.south west) to[bend left=50] ([shift={(-.1,.3)}]path picture bounding box.south east);\draw[black,-latex] ([shift={(0,.1)}]path picture bounding box.south) -- ([shift={(.3,-.1)}]path picture bounding box.north);}}}
\tikzset{roundnode/.append style={circle, draw=black, fill=gray!20, thick, minimum size=10mm}}
\tikzset{squarenode/.style={rectangle, draw=black, fill=none, thick, minimum size=10mm}}

\definecolor{Blues5seq1}{RGB}{239,243,255}
\definecolor{Blues5seq2}{RGB}{189,215,231}
\definecolor{Blues5seq3}{RGB}{107,174,214}
\definecolor{Blues5seq4}{RGB}{49,130,189}
\definecolor{Blues5seq5}{RGB}{8,81,156}

\definecolor{Greens5seq1}{RGB}{237,248,233}
\definecolor{Greens5seq2}{RGB}{186,228,179}
\definecolor{Greens5seq3}{RGB}{116,196,118}
\definecolor{Greens5seq4}{RGB}{49,163,84}
\definecolor{Greens5seq5}{RGB}{0,109,44}

\definecolor{Reds5seq1}{RGB}{254,229,217}
\definecolor{Reds5seq2}{RGB}{252,174,145}
\definecolor{Reds5seq3}{RGB}{251,106,74}
\definecolor{Reds5seq4}{RGB}{222,45,38}
\definecolor{Reds5seq5}{RGB}{165,15,21}

\definecolor{aquamarine}{rgb}{0.5, 1.0, 0.83}
\definecolor{babyblue}{rgb}{0.54, 0.81, 0.94}

\allowdisplaybreaks

\usepackage[most,breakable]{tcolorbox}
	{\expandafter\ifstrequal\expandafter{#1}{orange}{\begin{tcolorbox}[colback=red!15,colframe=orange!15,breakable,enhanced]}{\begin{tcolorbox}[colback=Blues5seq1,colframe=Blues5seq5,breakable,enhanced]}}%
	{\end{tcolorbox}}

	{\expandafter\ifstrequal\expandafter{#1}{orange}{\begin{tcolorbox}[colback=red!15,colframe=orange!15,breakable,enhanced]}{\begin{tcolorbox}[colback=white,colframe=babyblue,breakable,enhanced]}}%
	{\end{tcolorbox}}

\newcommand{\Err}{\operatorname{Err}}
\newcommand{\pg}{\operatorname{pg}}
\newcommand{\HS}[2]{\left\langle #1,#2\right\rangle}

\begin{document}
\title{Sharp pairwise reduction for quantum hypothesis testing}

\author{Kuan-Yi Lee}
\email{kuanyi.lee@sns.it}
\affiliation{Scuola Normale Superiore, Piazza dei Cavalieri 7, 56126 Pisa, Italy}

\author{Ludovico Lami}
\email{ludovico.lami@sns.it}
\affiliation{Scuola Normale Superiore, Piazza dei Cavalieri 7, 56126 Pisa, Italy}

\begin{abstract}
We determine the optimal universal coefficient in a reduction of multiple to binary quantum hypothesis testing. Specifically, for every finite ensemble in any Hilbert space dimension (finite or infinite), we prove that the error probability of the global pretty good measurement (PGM) is at most four times the sum of the optimal binary error probabilities. By constructing a family of regular-simplex ensembles, we further show that the coefficient four is optimal, even for arbitrary global measurements. This result improves on the pairwise bounds established by Cheng and Liu~\href{https://arxiv.org/abs/2606.06246}{[arXiv:2606.06246 (2026)]} and entails an explicit, sharp guarantee for the standard PGM itself. Our proof is also simpler: rather than constructing sequential measurements and applying a union bound, our proof relies on purely matrix-analytic techniques, combining a direct block Gram matrix analysis of the PGM error and a refined inequality between quantum Hellinger distance and Bures $\chi^2$-divergence. Our analysis also yields a refined, one-shot pairwise Chernoff bound and explicit sufficient number of copies for a desired discrimination error.
\end{abstract}
\maketitle

\textit{Introduction.}---Multiple quantum hypothesis testing aims to identify a quantum state drawn from a specified ensemble while minimizing the average error probability~\cite{Helstrom-paper,Yuen1975,Ogawa2000,Nussbaum2011,montanaro_pretty_2019}. For general ensembles, finding the minimum error requires an optimization over all measurements, and no general closed-form expression for this error or an optimal measurement is known. Although the optimization can be solved numerically by semidefinite programming~\cite{Eldar2003}, its computational complexity grows exponentially with the number of underlying qubits (and polynomially with the number of hypotheses)~\cite{Skrzypczyk2023_book}. Further, the optimal measurement will be, in general, difficult to implement experimentally. This motivates the identification of explicit, experimentally handy measurements whose performance can be evaluated analytically.

A natural approach to simplify the above optimization is to link this global discrimination problem to its binary sub-problems~\cite{qiu_minimum-error_2008,audenaert_upper_2014}. In contrast to the multiple-hypothesis setting, binary discrimination is characterized exactly by the Holevo--Helstrom formula~\cite{Helstrom-paper}, which expresses the optimal error in terms of trace norm~\cite{MARK}. A pairwise reduction would therefore turn the binary error into a quantitative estimate for the entire ensemble. However, such a reduction is not automatic -- the measurements that attains the different binary optima are generally different, and separately optimized binary tests do not directly specify a single valid measurement for the global task. This raises two natural questions: 
\begin{enumerate}
    \item[\textit{1.}] \textit{How tightly can the error of a multiple hypothesis test be controlled by its binary counterparts?}
    \item[\textit{2.}] \textit{Can the optimal reduction be achieved with an explicit, experimentally friendly measurement?}
\end{enumerate}

In this context, Audenaert and Mosonyi conjectured the existence of a pairwise reduction with a dimension-independent coefficient $C$~\cite[Conjecture 2.3]{audenaert_upper_2014}. Their numerical evidence further suggested $C=4$ as a valid candidate when the binary errors are summed without normalization. In the asymptotic setting, Ke Li proved that the optimal error exponent for a finite family of finite-dimensional states equals the smallest pairwise quantum Chernoff distance~\cite{Li2016}. More recently, Cheng and Liu improved the methods by Ke Li and established a dimension-independent one-shot bound with coefficient $C=8$ for positive trace-class operators~\cite[Theorem 3.1]{Cheng2026}. Their proof adapts sequential binary tests on the rank-one spectral components of the hypotheses, followed by a quantum union bound on such events~\cite{Gao2015}. Although it resolves the existence conjecture by Audenaert and Mosonyi, the construction itself is experimentally difficult to realize. These studies leave open the optimal universal coefficient and which explicit measurement achieves the corresponding optimal guarantee.

In this work, we determine the \textit{optimal} universal coefficient, and furthermore show that it can be attained by the global pretty good measurement~\cite{hausladen_pretty_1994} (PGM) in all quantum systems, finite- or infinite-dimensional. Specifically, we prove that the PGM error is at most four times the sum of the optimal binary errors, and a regular-simplex ensemble further shows that no smaller coefficient can hold uniformly over the number of hypotheses, even when optimizing the global measurement.

The centrality of the PGM in quantum hypothesis testing stems first and foremost from the celebrated Barnum--Knill theorem~\cite{BarnumKnill}, which states that the associated error probability differs from that of the optimal measurement by a factor of at most two. 
This is important from the technological point of view, because the PGM is often experimentally much easier to implement than an arbitrary measurement~\cite{pretty-good}. Our findings add a further dimension to this centrality of the PGM in the theory of quantum hypothesis testing, by showing that the PGM error probability, which may be challenging to calculate for large ensembles, can be tightly bounded by looking at the binary reductions. This paves the way for further applications of this key primitive of quantum state discrimination.

The key ingredient in our proof is a constant-one inequality between squared quantum Hellinger distance and Bures $\chi^2$-divergence, improving the the coefficient obtained by combining existing inequalities~\cite[Fact 2.25 and Proposition 2.31]{Flammia2024}. In addition, the same argument further yields a refined one-shot pairwise Chernoff bound for the global PGM. Applied to multiple, independent copies, this bound gives a sufficient number of copies for discrimination at a prescribed error tolerance.\smallskip

\textit{Preliminaries.}---We first briefly recap the setting of quantum hypothesis testing. Consider two quantum states $\rho$ and $\sigma$, with the corresponding priors $p,q$, a binary test $0\le \Pi\le \id$ with corresponding outcome to decision whether the state is $\rho$ or $\sigma$. The optimal error of such a task is given by the Holevo--Helstrom formula~\cite{Helstrom-paper}:
\bb
    \Err^\star(p\rho,q\sigma)
    &\coloneqq 
    \min_{0\le \Pi\le\id}
    \left\{
        p\Tr \rho(\id-\Pi)+ q \Tr \sigma \Pi
    \right\} \\
    &=
    \frac12 \left( p + q -\|p\rho-q\sigma\|_1 \right),
\ee
where $\|X\|_1\coloneqq \Tr\sqrt{X^\dagger X}$ is the trace norm. For a finite ensemble $\{p_i,\rho_i\}_{i=1}^N$, where the $p_i$ are prior probabilities and the $\rho_i$ are density operators, we absorb the priors into $A_i\coloneqq p_i\rho_i$ for convenience. The corresponding optimal error can then be written as
\bb
    \Err^\star(A_1,\ldots,A_N)
    \coloneqq 
    \min_{\{M_i\}}
    \sum_{i=1}^N \Tr A_i(\id-M_i),
    \label{eq:optimal-error}
\ee
with POVM $\{M_i\}_i$ such that $M_i\ge 0$ and $\sum_{i=1}^N M_i=\id$. 
% We note that all of our bounds are homogeneous and also apply to arbitrary positive semi-definite $A_i$.

Apart from optimizing over all POVMs, we consider the pretty good measurement (PGM), which is determined directly by the ensemble: for a given ensemble $\{A_i\}_{i=1}^N$, set $S\coloneqq \sum_{i=1}^N A_i$; the PGM is then defined by $M_i^{\pg} \coloneqq  S^{-1/2}A_i\, S^{-1/2}$, where $S^{-1/2}$ denotes the inverse square root of the restriction of $S$ to its support, and satisfies $\sum_i M_i^{\pg} = \id$ on $\supp S$. The PGM error is
\bb
    \Err^{\pg}(A_1,\ldots,A_N)
    &\coloneqq 
    \sum_{i=1}^N\Tr A_i\left(\id -M^{\pg}_i\right)\\
    &=
    \sum_{i\neq j} \Tr A_i S^{-1/2}A_j S^{-1/2}.
    \label{eq:err_PGM}
\ee
For simplicity, we use shorthand notation for $N$-ensemble (global) errors as $\Err^\star_N$ and $\Err^{\pg}_N$, and, for pairwise error, $\Err^\star_{ij}\coloneqq \Err^\star(A_i,A_j)$.
% and $\Err^{\pg}_{ij}\coloneqq \Err^{\pg}(A_i,A_j)$. 
We note that, in particular, $\Err^\star(A_i,A_j)$ is \textit{not} a conditionally normalized binary error; rather, it includes the original priors $p_i,p_j$.\smallskip

\textit{Sharp pairwise reduction.}---Our main result controls the error associated to a global measurement with the sum of its binary errors.

\begin{thm}[(Pairwise reduction)]
\label{thm:sharp-pairwise}
For every finite ensemble $A_1,\ldots,A_N \ge 0$ in finite or infinite dimension, the PGM error satisfies
\bb
    \Err^{\star}_N
    \le
    \Err^{\pg}_N
    \le
    4\sum_{i< j}\Err^{\star}_{ij}.
\ee
The coefficient four is uniformly optimal.
\end{thm}
Remark that the coefficient is independent of number of hypotheses, priors, and Hilbert space dimension; its sharpness is guaranteed by a regular-simplex ensemble. The trace-class extension is a corollary of Proposition~\ref{prop:truncation}.

Theorem~\ref{thm:sharp-pairwise} improves the coefficient $C=8$ in the pairwise bound of Cheng and Liu~\cite[Theorem 3.1]{Cheng2026} to $C=4$ in trace-class operators and gives the bound for the standard PGM itself. It also confirms the value suggested by numerical evidence of Audenaert and Mosonyi~\cite[Conjecture 2.3]{audenaert_upper_2014}.
Unlike the construction in Ref.~\cite{Cheng2026}, which relies on sequential measurements combined via a quantum union bound~\cite{Gao2015,ODonnell2022}, 
our proof directly analyzes the PGM through its block Gram matrix representation, retaining a more transparent route and a simpler yet stronger result.

Importantly, this bound is achieved by the standard PGM measurement, which is widely used in quantum information theory as a near-optimal estimation~\cite{BarnumKnill} and also has concrete optical realizations~\cite{Clarke2001,PhysRevLett.118.100501}. 
Another remarkable feature is that the global PGM error is directly controlled by the optimal binary errors, so that the binary errors provide a common benchmark for both the optimal global measurement and the PGM, with the best possible universal coefficient.

Moreover, this bound also yields finite-copy guarantees for the global PGM on an $n$-copy ensemble. We derive explicit sufficient copy numbers in terms of pairwise trace overlaps after Theorem~\ref{thm:pairwise-Chernoff}.
\smallskip

\textit{Proof structure.}---The proof has two complementary parts. We first show that there is no universal coefficient small than four is possible by using a regular-simplex ensemble -- both their binary and global minimum errors can be evaluated explicitly. Then, second, we prove that four suffices for arbitrary ensembles by a block Gram matrix analysis together with two operator inequalities, which we introduce after the sharpness argument.
\begin{proof}[Shapness in Theorem~\ref{thm:sharp-pairwise}]
Let $\{\ket{v_i}\}_{i=1}^N\subset\mathbb C^{N-1}$ be unit vectors forming a regular simplex, i.e., $\braket{v_i|v_j} = -1/(N-1)$ for $i\neq j$, and set $A_i = \ket{v_i}\bra{v_i}/N$. 
Then, the Holevo--Helstrom formula gives the optimal binary error:
\bb
    \Err^{\star}_{ij}=\frac 1N \left( 1-\sqrt{1-\frac{1}{(N-1)^2}} \right).
\ee

For the $N$ ensemble case, consider the dual form of its optimal success probability, namely,
\bb
    \operatorname{Succ}^\star_N = \min \left\{\Tr Y : Y \ge A_i ~~\forall i, \quad Y= Y^\dagger\right\}.
\ee
Therefore, any $Y\ge A_i$ gives an upper bound on $\operatorname{Succ}^\star_N$. By the property of simplex $\sum_{i}A_i = \id/(N-1)$, we have
\bb
    \sum_i M^{\pg}_i A_i 
    = 
    % \sum_i (N-1)A_i^2 
    % =
    % \sum_i \frac{N-1}{N} A_i 
    % = 
    \frac 1N \id \, \ge \, A_i.
\ee
Then, we simply take $Y = \sum_i M^{\pg}_i A_i$ and obtain
\bb
    \operatorname{Succ}^\star_N 
    \le 
    \Tr Y 
    =
    \operatorname{Succ}^{\pg}_N = \frac{N-1}{N}.
\ee
By reverse, we have $\operatorname{Succ}^{\pg}_N\le \operatorname{Succ}^\star_N $; hence $\Err^{\star}_N=1/N$.

Any coefficient $C$ valid uniformly for the optimal global error must satisfy
\bb
    C 
    \ge 
    \frac{\Err^\star_N}{\sum_{i<j}\Err^\star_{ij}} 
    = 
    \left( 1+\sqrt{1-\frac{2}{N}} \right)^{2} \xrightarrow[N\to \infty]{} 4.
\ee
Thus, no universal factor can strictly smaller than 4.
\end{proof}
The same family therefore establishes sharpness for both the optimal global error and the PGM error. Sharpness is only achieved in the limit $N\to\infty$; we leave the determination of the optimal coefficient for each fixed $N$ as an open problem.
\smallskip

\textit{The universal upper bound.}---Having shown that four
is necessary, we now prove that it suffices for the global PGM. The main difficulty in directly analyzing the global PGM is the term $S^{-1/2}$ that depends on all hypotheses simultaneously. A block Gram representation rephrases the PGM error as a sum of Hilbert--Schmidt norms of off-diagonal blocks of a matrix square root, as earlier in Ref.~\cite{montanaro_pretty_2019}. Although each block of the Gram matrix involves only a pair of hypotheses, its square root still composes the entire ensemble. This motivates an upper bound by a form whose weight depends only on the block-diagonal part of the Gram matrix.

Moreover, the Hilbert--Schmidt norm also makes scalar comparisons available. The left- and right-multiplication superoperators commute even when the input operators do not. Under their joint basis, which also underlies the Petz quasi-entropy~\cite[Section 2]{Hiai2011}, allows several key operator inequalities below to be reduced to pointwise scalar inequalities.

Two Propositions are therefore introduced to complete this strategy. First, Proposition~\ref{prop:H-chi-bound} bounds the Hilbert--Schmidt norm by a quadratic form with left- and right-multiplication superoperators determined by the reference operator. Choosing the block-diagonal part of the Gram matrix as the reference makes the associated superoperator act independently on each matrix block, yielding harmonic terms that depend only on individual pairs of hypotheses. 
Second, thereafter, Proposition~\ref{prop:pgm-Helstrom} then bounds each harmonic term by the corresponding optimal binary error, yielding the desired reduction. We establish these Propositions before combining them to prove the upper bound in Theorem~\ref{thm:sharp-pairwise}.

\begin{prop}[(Hellinger--$\chi^2$ inequality)]
\label{prop:H-chi-bound}
Let \(X,Y\ge0\) and define $\Omega_Y\coloneqq (L_Y+R_Y )/2$ with left $L_Y(X)=YX$ and right $R_Y(X)=XY$ multiplications, we have
\bb
    \left\|X^{1/2}- Y^{1/2}\right\|_2^2
    \le
    \HS{X-Y}{\Omega_Y^{-1}(X-Y)},
    \label{eq:H-chi-bound}
\ee 
where $ \|X\|_2^2\coloneqq \HS{X}{X}$ with $\HS{X}{Y}\coloneqq \Tr X^\dagger Y$.
\end{prop}
For density operators, the two sides are respectively the squared quantum Hellinger distance and the Bures $\chi^2$-divergence. The coefficient-one inequality improves the coefficient two obtained by combining Fact 2.25 and Proposition 2.31 of Ref.~\cite{Flammia2024}, and it remains valid for positive semidefinite operators. In the block Gram matrix analysis, removing this factor-two loss is essential to obtaining the sharp coefficient four.

\smallskip
\begin{proof}[Proof of Proposition~\ref{prop:H-chi-bound}]
Choose orthonormal eigenbases $E^X_i=\ket{x_i}\bra{x_i},~E^Y_j=\ket{y_j}\bra{y_j}$ and set
\(
    X= \sum_i x_i E_i^X,~
    Y= \sum_j y_j E_j^Y,~
    c_{ij}\coloneqq  \Tr E_i^X E_j^Y = |\braket{x_i|y_j}|^2,
\)
where $x_i\ge 0$ and $y_j>0$. Consider the superoperator $(L_X+R_Y)$ acts on a joint basis $E_{ij}^{XY}\coloneqq |x_i\rangle\langle y_j|$, we have
\bb
    &(L_X+R_Y)\big(E_{ij}^{XY}\big) = (x_i+y_j)\,E_{ij}^{XY}.
    \label{eq:joint-spectrum}
\ee
Given that $(X-Y)_{ij}=(x_i-y_j)\braket{x_i|y_j}$ under this basis, we can therefore consider a scalar inequality
\(
    (\sqrt x-\sqrt y)^2
    \le
    (x-y)^2/(x+y),
\)
which gives
\bb
    \left\|X^{1/2}-Y^{1/2}\right\|^2_2 
    =
    \sum_{ij}(\sqrt{x_i}-\sqrt{y_j})^2c_{ij}
    \le 
    \sum_{ij}\frac{(x_i-y_j)^2}{x_i+y_j}c_{ij}.
    % \\
    % &\le
    % \HS{X-Y}{(L_X + R_Y)^{-1}(X-Y)}.
    \label{eq:rhs-1-ineq-for-prop-hc}
\ee

On the other hand, by linearity and $\Omega_Y^{-1}(Y)=\id$, the RHS of Eq.~\eqref{eq:H-chi-bound} can be written as
\bb
    \big\langle X-Y,& \,\Omega_Y^{-1}(X-Y)\big\rangle\\
    &=
    % \HS{X}{\Omega_Y^{-1}X} - 2\Tr X +\Tr Y\\
    % &=
    \sum_{ik}x_i x_k\Tr E_i^X\Omega_Y^{-1}E_k^X - 2\sum_i x_i +\sum_{ij} y_{j}c_{ij}\\
    &\ge
    \sum_i x_i^2\Tr E_i^X\Omega_Y^{-1}E_i^X - 2\sum_i x_i +\sum_{ij} y_{j}c_{ij}\\
    &=
    \sum_i \left( x_i^2\sum_{jk}\frac{2c_{ij}c_{ik}}{y_j+y_k} -2x_i +\sum_j y_j c_{ij} \right) \\
    &\ge
    \sum_{ij}\frac{(x_i-y_j)^2}{x_i+y_j}c_{ij},\label{eq:prop-for-appendix}
\ee
where the first inequality is because $\Omega_Y^{-1}$ is positive preserving (see the Appendix~\ref{app:positive-preserving}), we can discard all $k\ne i$ term to get the inequality, and the second (scalar) inequality is given in Appendix~\ref{app:scalar-ineq}.
Combining Eqs.~\eqref{eq:rhs-1-ineq-for-prop-hc} and \eqref{eq:prop-for-appendix} proves the claim.
\end{proof}

As an immediate consequence, any estimator with Bures $\chi^2$-error at most $\varepsilon$ with probability at least $1-\delta$ also has squared Hellinger error at most $\varepsilon$ with the same confidence guarantee. Thus, Proposition~\ref{prop:H-chi-bound} removes a factor-two loss in converting between these criteria; see Ref.~\cite[Theorem~1.6 and Corollary~3.24]{Flammia2024} for $\chi^2$-tomography guarantees.

We next compare the blockwise term directly with the optimal binary error. The argument uses a block row superoperator on the Hilbert--Schmidt space.

\begin{prop}[(Link to binary errors)]
\label{prop:pgm-Helstrom}
Let $X,Y\ge 0$, then
\bb
    \HS{X^{1/2}Y^{1/2}}{(L_X+ R_Y)^{-1}(X^{1/2}Y^{1/2})}
    \le
    \Err^\star(X,Y).
    \label{eq:pgm-Helstrom}
\ee
\end{prop}
The above inequality is the unweighted harmonic~\cite{Kubo1980} lower bound of Ref.~\cite[Proposition~B.1]{Cheng2026}. Its role here is to connect the pairwise quadratic terms arising from the block Gram analysis to the binary error. We give a simple and direct inner product proof in the following.

\begin{proof}[Proof of Proposition~\ref{prop:pgm-Helstrom}]
Let $\Pi^\star\coloneqq \{X\ge Y\}$ be the Holevo--Helstrom projection, such that
\(
    \Err^\star(X,Y)
    =
    \Tr X (\id-\Pi^\star)+\Tr Y\Pi^\star.
\)
Define the block row superoperator as
\bb
    \mathcal T
    \coloneqq 
    \begin{pmatrix}
        L_{X^{1/2}} & R_{Y^{1/2}}
    \end{pmatrix},
\ee
such that $\mathcal T\mathcal T^\dagger = L_X + R_Y$. For a block column $W=\Pi^\star Y^{1/2}\oplus X^{1/2}(\id -\Pi^\star)$, we have
\bb
    \mathcal T W
    &=
    X^{1/2}\Pi^\star Y^{1/2} + X^{1/2}(\id -\Pi^\star)Y^{1/2}\\
    &=
    X^{1/2}Y^{1/2}.
\ee
Now, we can rewrite the LHS of Eq.~\eqref{eq:pgm-Helstrom} in terms of $\mathcal T, W$, and $ \mathcal T^\dagger\big(\mathcal T\mathcal T^\dagger\big)^{-1}\mathcal T$ is a self-adjoint projector, implying
\bb
    0
    \le
    \mathcal T^\dagger\big(\mathcal T\mathcal T^\dagger\big)^{-1}\mathcal T
    \le
    \id.
\ee
Acting on left (right) side with $W^\dagger$ ($W$), respectively, we obtain
\(
    \Tr \big(\mathcal TW\big)^\dagger\big(\mathcal T\mathcal T^\dagger\big)^{-1}\big(\mathcal T W\big)
    \le
    \Tr W^\dagger W
\)
and get
\bb
    \HS{\mathcal T W}{\big(\mathcal T\mathcal T^\dagger\big)^{-1}\mathcal T W}
    &\le 
    \left\|W \right\|^2_2\\
    &= 
    \left\|\Pi^\star Y^{1/2} \right\|^2_2 + \left\|X^{1/2}(\id -\Pi^\star) \right\|^2_2\\
    &=
    \Tr Y\Pi^\star + \Tr X(\id-\Pi^\star)\\
    &= \Err^\star(X,Y).
\ee
As the claim, we conclude the proof.
\end{proof}

Now we are ready for the proof of Theorem~\ref{thm:sharp-pairwise}.

\begin{proof}[Proof of Theorem~\ref{thm:sharp-pairwise}]
Define $D\coloneqq \diag(A_1,\ldots,A_N)$ as a block diagonal matrix and 
\(
    T = (A_1^{1/2}A_2^{1/2}\! \ldots \,A_N^{1/2})
\).
Then $TT^\dagger=S$, while $G\coloneqq T^\dagger T$ is the block Gram matrix with $G_{ij}=A_i^{1/2}A_j^{1/2}$. The polar decomposition gives $T^\dagger S^{-1/2}T= (T^\dagger T)^{1/2} =G^{1/2}$. Taking the $(i,j)$ block of $G^{1/2}$, Eq.~\eqref{eq:err_PGM} gives
\bb
    \Err^{\star}_N \le \Err^{\pg}_N
    &=
    \sum_{i\neq j} \left\|A_i^{1/2} S^{-1/2}A_j^{1/2}\right\|^2_2\\
    &= 
    \sum_{i\neq j} \left\|\big(G^{1/2} - D^{1/2} \big)_{ij}\right\|_2^2 \\
    &\le 
    \sum_{i j} \left\|\big(G^{1/2} - D^{1/2} \big)_{ij}\right\|_2^2\\
    &= 
    \left\|G^{1/2} - D^{1/2} \right\|_2^2,
    \label{eq:}
\ee
where $\|X\|_2^2\coloneqq \HS{X}{X}$ and $\HS{X}{Y}\coloneqq \Tr X^\dagger Y$ denotes as the Hilbert--Schmidt inner product. Notice that the second line is because $D^{1/2}$ is block diagonal; the third line is given by the non-negative terms of $i=j$ block. Proposition~\ref{prop:H-chi-bound}, the Hellinger--$\chi^2$ inequality, gives
\bb
    \left\|G^{1/2}-D^{1/2}\right\|_2^2 
    \le
    \HS{G-D}{\Omega_D^{-1}(G-D)},
\ee
where the superoperator $\Omega_D$ acts blockwise. On the $i\neq j$ block, $(G-D)_{ij} \to A_i^{1/2}A_j^{1/2}$ and $\Omega_D \to (L_{A_i}+R_{A_j})/2$. It follows that
\bb
    \big\langle G-D &,\Omega_D^{-1}(G-D)\big\rangle \\
    &=
    2\sum_{i\neq j}
    \HS{A_i^{1/2}A_j^{1/2}}{(L_{A_i}+R_{A_j})^{-1}(A_i^{1/2}A_j^{1/2})}\\
    &\le 
    4\sum_{i< j} 
    \Err^\star(A_i,A_j),
    \label{eq:main-proof}
\ee
where the last line is the Proposition~\ref{prop:pgm-Helstrom}. 
\end{proof}
\smallskip

\textit{One-shot Chernoff bound.}---The harmonic quantity, i.e.\ LHS of Eq.~\eqref{eq:pgm-Helstrom},  provides an intermediate estimate for the global PGM error. A further upper bound in terms to the optimal binary errors gives the sought sharp pairwise reduction. Alternatively, bounding the harmonic quantity directly by trace overlaps leads to the one-shot Chernoff distance below; the same argument also yields a sharp coefficient.

\begin{thm}[(One-shot Chernoff distance)]
\label{thm:pairwise-Chernoff}
For every $A_i,A_j\ge 0$ and $s_{ij}\in(0,1)$, it holds that
\bb
    \Err^{\pg}_N
    \le 
    4\sum_{i<j} \alpha(s_{ij}) \Tr A_i^{1-s_{ij}}A_j^{s_{ij}},
\ee
where $\alpha(s) \coloneqq  s^{s}(1-s)^{1-s} \in [0.5,1)$. The coefficient $\alpha(s_{ij})$ is optimal for all pairs $i,j$ with $i<j$.
\end{thm}

In the finite-dimensional one-shot case, Theorem~\ref{thm:pairwise-Chernoff} strengthens the Chernoff branch of~\cite[Theorem~3.1]{Cheng2026} by an additional factor $\alpha(s_{ij})\in [0.5,1)$ in each term, and in addition it controls the error of the global PGM rather than merely that of the globally optimal measurement.

\begin{proof}
Expanding $\HS{X^{1/2}Y^{1/2}}{(L_X+ R_Y)^{-1}X^{1/2}Y^{1/2}}$ and $\Tr X^{1-s}Y^s$ in the basis $E^{XY}_{ij}$, as we did before, we have
\bb
    \sum_{ij}\frac{x_iy_j}{x_i+y_j}c_{ij}
    \le C'
    \sum_{ij} x_i^{1-s}y_j^s c_{ij},
    \label{eq:Chernoff-constant}
\ee
where $C'$ is the constant we wish to find. Now, for a fixed $(i,j)$, dividing both sides by $x^{1-s}y^s$ and changing variable to $t=x/y$, we obtain a concave function with maximum
\bb
    \sup_{t>0} \frac{t^s}{t+1} = (1-s)^{1-s}s^s ,
\ee
attained when $t=s/(1-s)$. Then, taking $C'=(1-s)^{1-s}s^s$ and applying it to Eq.~\eqref{eq:Chernoff-constant}, we conclude the proof.
\end{proof}

As a clear application of Theorem~\ref{thm:pairwise-Chernoff}, we consider a fixed ensemble $\{p_i,\rho_i\}_{i=1}^N$ and let $\Err_N^{\pg}(n)$ denote the error of the PGM associated with $\{p_i,\rho_i^{\otimes n}\}_{i=1}^N$. Taking $s_{ij}=1/2~~\forall i,j$ and setting $\epsilon\coloneqq\max_{i<j}\Tr\sqrt{\rho_i}\sqrt{\rho_j}$, Theorem~\ref{thm:pairwise-Chernoff} gives
\bb
    \Err_N^{\pg}(n) 
    \le 
    2\sum_{i<j}\sqrt{p_ip_j}\, \epsilon^{n}
    \le 
    (N-1)\epsilon^{n}.
    \label{eq:finite-copy}
\ee
Therefore, it suffices to take
\bb
    n\ge
    \left\lceil \frac{\log(N-1) - \log \delta}{\log (1/\epsilon)} \right\rceil
\ee
to ensure the global PGM error at most $\delta \in (0,1)$.
The same estimate for the optimal error follows from Ref.~\cite[Theorem 3.5]{audenaert_upper_2014}; here, it holds directly for the global PGM. Furthermore, it also improves the error in the Chernoff branch of Ref.~\cite[Theorem 3.1]{Cheng2026}, corresponding to an additive reduction of $\log 2/\log(1/\epsilon)$ in the sufficient copy threshold.
Note that the pairwise sum in Eq.~\eqref{eq:finite-copy} can be used directly when the individual overlaps are available, avoiding the replacement by the worst-case pair, while the statement itself serves as a finite budget estimate. \smallskip

\textit{Trace-class extension.}---The preceding results and inequalities are independent of the Hilbert space dimension, suggesting that they should be valid also for infinite-dimensional systems. 
To justify this, we now show the following continuity statement.

\begin{prop}[(Trace-class extension)]
\label{prop:truncation}
In an infinite-dimensional separable Hilbert space, consider an orthonormal basis $\{\ket{k}\}_{k=1}^\infty$ and a finite-cutoff projector $\Pi^m \coloneqq \sum_{k=0}^m \ket{k}\bra{k}$. 
%with $\Pi^m \uparrow \id$. 
Then, for a finite collection of trace-class operators $(A_i)_i$, let $A_i^{(m)}\coloneqq \Pi^m A_i\,\Pi^m$ be the $m$-level truncation of the $i$-th operator. Then
%such that $\lim_{m\to \infty}\|A_i^{(m)} - A_i\|_1=0$, we have
\bb
    \lim_{m\to \infty} \Err^{\star}\big(A_1^{(m)},\ldots,A_N^{(m)}\big)
    =
    \Err^{\star}(A_1,\ldots,A_N).
\ee
\end{prop}
\begin{proof}
See Appendix~\ref{app:truncation}
\end{proof}

The trace-class extension 
extends our binary reduction to finite ensembles on bosonic Fock spaces, including Gaussian states with infinite-dimensional support~\cite{weedbrook12}. Consequently, upper bounds on the corresponding binary minimum errors~\cite{pirandola_computable_2008} yield an upper bound on the global minimum error with the same sharp coefficient. This includes discrimination of a finite optical ensemble without imposing photon-number truncations. \smallskip

\textit{Discussion.}---We have established a sharp universal relation between multiple quantum hypothesis testing and its binary sub-problems. In any Hilbert space dimension, the PGM achieves the coefficient four, which we have proved to be optimal for all global measurements.  
We should interpret this as an optimal universal guarantee, rather than a statement that the PGM is optimal for every given ensemble.

Operationally speaking, we have benchmarked the performance of the globally optimal measurement by using its binary reductions.
The accompanying Chernoff bound makes this benchmark directly usable for finite-copy testing problems, providing a sufficient number of copies in terms of pairwise trace overlap.
In addition, the improved Hellinger--$\chi^2$ inequality turns the Bures $\chi^2$-divergence into Hellinger distance without an additional loss, making an available link in estimations.

The developments in this work suggest some natural open questions.  
First, it should be possible to control the full $N$-hypothesis problem through a subset of $k$ hypotheses. Inspired by the regular-simplex ensembles, we conjecture that the optimal coefficients, corresponding to the reduction from $N$ to $k$, of the normalized $k$-wise reduction are non-increasing in $k$ (see Appendix~\ref{app:PGM-pairwise}). Such a hierarchy would allow the subset size to be chosen according to the available sub-ensemble error estimates, potentially sharpening the sufficient number copies under a prescribed error tolerance.
Second, a separate direction is to extend the pairwise reduction to the experimentally relevant family of continuous ensembles. In such a setting, a score function~\cite[Section 3.2]{mishra_near-optimal_2025} is required to connect the pairwise bounds to the estimation of continuous labels, including optical displacements or phases~\cite{weedbrook12}.

\subsection*{Acknowledgements} 
We sincerely thank Hao-Chung Cheng for the invaluable discussion.
K.-Y. L and LL acknowledge financial support from the European Union (ERC StG ETQO, Grant Agreement no.\ 101165230). 

ChatGPT-5.5 and earlier versions were used for literature review and numerical searches for counterexamples and candidate inequalities. These searches helped motivate the universal constant conjecture and inspired some of the proof ideas. GPT-5.6 Sol was also used for proof checking, related-work comparison, and language polishing. All mathematical content was independently verified by the authors.

% \bibliography{ref.bib}

\newpage
\appendix
\begin{widetext}

\section{Proof of positive preserving}
\label{app:positive-preserving}
In the following, we use Sylvester integral representation to show the superoperator $\Omega_Y^{-1}$ is positive preserving.
We note that $\Omega_Y(E_{ij}^Y)=(y_i+y_j)E_{ij}^Y/2$, with $\Omega_Y = (L_Y + R_Y)/2$ and get
\bb
    \Omega^{-1}_Y(E^X_i)
    &=
    \Omega^{-1}_Y\left(\sum_{jk}\bra{y_j}E^X_i\ket{y_k}E^Y_{jk}\right)\\
    &=
    \sum_{jk} \frac{2\bra{y_j}E^X_i\ket{y_k}}{y_j+y_k}E^Y_{jk}\\
    &= 
    \sum_{jk} 2\bra{y_j}E^X_i\ket{y_k}\int_0^\infty e^{-t y_j}E^Y_{jk}e^{-t y_k}\,dt\\
    &= 2\int_0^\infty e^{-tY}E^X_{i}e^{-tY}\,dt \ge 0.
\ee
Positivity follows because each map $E^X_i \mapsto e^{-tY}E^X_i e^{-tY}$ is positive.

\section{The scalar inequality in the proof of Proposition~\ref{prop:H-chi-bound}}
\label{app:scalar-ineq}
Let $y_j>0$, $c_j\geq0$, and $\sum_jc_j=1$. For every $x\geq0$,
\bb
    x^2 \sum_{jk}\frac{2 c_j c_k}{y_j+y_k} -2x +\sum_jy_jc_j
    \geq
    \sum_j\frac{(x-y_j)^2}{x+y_j}c_j.
    \label{eq:scalar-lemma}
\ee

\begin{proof}
The case $x=0$ is an equality. For $x>0$, the forth line in Eq.~\eqref{eq:scalar-lemma} is equivalent to
\begin{equation}
    1+x\sum_{jk}\frac{2c_jc_k}{y_j+y_k}
    \geq
    4\sum_j\frac{x}{x+y_j}c_j.
    \label{eq:scalar-ineq}
\end{equation}
Set $f(t)\coloneqq \sum_j c_j e^{-t y_j}$, then we have
\bb
    1+x\sum_{jk}\frac{2 c_j c_k}{y_j+y_k}
    &=
    1+2x\int_0^\infty f(t)^2\,dt,\\
    4\sum_j\frac{x}{x+y_j}c_j
    &=
    4x\int_0^\infty e^{-xt}f(t)\,dt.
\ee
Applying AMGM inequality
\(
    2e^{-xt}f(t)
    \leq
    e^{-2xt}+f(t)^2
\)
implies
\bb
    4x\int_0^\infty e^{-xt}f(t)\,dt
    &\le
    2x\int_0^\infty \left[ e^{-2xt}+f(t)^2 \right]\,dt\\
    &=
    1+2x\int_0^\infty f(t)^2\,dt,
\ee
which proves Eq.~\eqref{eq:scalar-ineq}.
\end{proof}

\section{Proof of trace-class extension}
\label{app:truncation}

Let us consider an orthonormal basis $\{\ket{i}\}_i$ and define
\bb
    \Pi^m \coloneqq \sum_i^m \ket{i}\bra{i},\qquad \Pi^m \uparrow \id.
\ee
Then, for every $A_i\ge 0$ be a trace-class operator in separable Hilbert space, we let $A_i^{(m)} \coloneqq \Pi^m A_i\,\Pi^m$ such that $\lim_{m\to \infty}\|A_i^{(m)} - A_i\|_1=0$, we have
\bb
    \lim_{m\to \infty} \Err^{\star}(A_1^{(m)},\ldots,A_N^{(m)})
    =
    \Err^{\star}(A_1,\ldots,A_N)
\ee
Here, a finite rank projector $\Pi^{(m)}$ selects a subspace in a finite dimension approximation of the hypothesis $\{A_i\}$. Physically, for instance, such a projector in a single bosonic mode is the Fock basis corresponding to impose a photon number truncation.
\begin{proof}
Fix a POVM $\{M_i\}$, by the property of trace norm~\cite[Lemma 9.5]{NielsenChuang2010}, we obtain
\bb
    \left| \sum_i \Tr A_i (\id - M_i) - \sum_i A^{(m)}_i  (\id - M_i) \right| 
    &=
    \left| \sum_i \Tr \left(A_i-A^{(m)}_i\right)(\id - M_i) \right| \\
    &\le
    \sum_i \left| \Tr \left(A_i-A^{(m)}_i\right )(\id - M_i) \right| \\
    &\le
    \sum_i \left\| A_i-A^{(m)}_i\right\|_1.
    \label{eq:stable}
\ee
Given that Eq.~\eqref{eq:stable} holds $\forall \{M_i\}$, we now let $\varepsilon>0$ such that $\Err^\star(A_1,\ldots,A_N)\ge \sum_i \Tr A_i (\id - M_i)-\varepsilon$; then, we have
\bb
    \Err^\star\left(A^{(m)}_1,\ldots,A^{(m)}_N\right) 
    &\le
    \sum_i \Tr A^{(m)}_i (\id - M_i)\\
    &\le
    \sum_i \Tr A_i (\id - M_i) + \sum_i \left\| A_i-A^{(m)}_i\right\|_1 \\
    &\le
    \Err^\star(A_1,\ldots,A_N) + \varepsilon + N\max_i \left\| A_i-A^{(m)}_i\right\|_1,
\ee
by taking $m\to \infty$ and choose the optimal $\{M_i\}$ such that $\varepsilon \downarrow 0$ for the ensemble $\{A_i\}$, we get
\bb
    \lim_{m\to \infty}\Err^\star\left(A^{(m)}_1,\ldots,A^{(m)}_N\right) \le \Err^\star(A_1,\ldots,A_N).
    \label{eq:epsilon}
\ee
By choosing $\{M_i'\}$ such that $\Err^\star\big(A_1^{(m)},\ldots,A_N^{(m)}\big)\ge \sum_i \Tr A_i^{(m)} (\id - M_i')-\varepsilon'$ and $\varepsilon'\downarrow 0$, we obtain an opposite direction of the inequality in Eq.~\eqref{eq:epsilon}, and hence
\bb
    \lim_{m\to \infty}\Err^\star\left(A^{(m)}_1,\ldots,A^{(m)}_N\right) = \Err^\star(A_1,\ldots,A_N).
\ee
\end{proof}

\begin{proof}
Fix a POVM $\{M_i\}$, we obtain
\bb
    \Big| \sum_i \Tr A_i (\id - M_i) &- \sum_i A^{(m)}_i  (\id - M_i) \Big| \\ 
    &=
    % \Big| \sum_i \Tr \left(A_i-A^{(m)}_i\right)(\id - M_i) \Big| \\
    % &\le
    \sum_i \Big| \Tr \left(A_i-A^{(m)}_i\right )(\id - M_i) \Big| \\
    &\le
    \sum_i \left\| A_i-A^{(m)}_i\right\|_1,
    \label{eq:stable}
\ee
where the last line is by the dual of trace norm.
On the one hand, given that Eq.~\eqref{eq:stable} holds $\forall \{M_i\}$, we now let $\varepsilon>0$ such that $\Err^\star(A_1,\ldots,A_N)\ge \sum_i \Tr A_i (\id - M_i)-\varepsilon$; then, we have
\bb
    \Err^\star\big(A^{(m)}_1&,\ldots,A^{(m)}_N\big) \\
    &\le
    \sum_i \Tr A^{(m)}_i (\id - M_i)\\
    &\le
    \sum_i \Tr A_i (\id - M_i) + \sum_i \left\| A_i-A^{(m)}_i\right\|_1 \\
    &\le
    \Err^\star(A_1,\ldots,A_N) + \varepsilon + N\max_i \left\| A_i-A^{(m)}_i\right\|_1,
\ee
by taking $m\to \infty$ and choose the optimal $\{M_i\}$ such that $\varepsilon \downarrow 0$ for the ensemble $\{A_i\}$, we get
\bb
    \lim_{m\to \infty}\Err^\star\left(A^{(m)}_1,\ldots,A^{(m)}_N\right) \le \Err^\star(A_1,\ldots,A_N).
    \label{eq:epsilon}
\ee
On the other hand, by choosing $\{M_i'\}$ such that $\Err^\star\big(A_1^{(m)},\ldots,A_N^{(m)}\big)\ge \sum_i \Tr A_i^{(m)} (\id - M_i')-\varepsilon'$ and $\varepsilon'\downarrow 0$, we obtain an opposite direction of the inequality in Eq.~\eqref{eq:epsilon}, and hence
\bb
    \lim_{m\to \infty}\Err^\star\left(A^{(m)}_1,\ldots,A^{(m)}_N\right) = \Err^\star(A_1,\ldots,A_N),
\ee
as the claim.
\end{proof}

\section{Pretty-good-measurement pairwise reduction and regular-simplex examples}
\label{app:PGM-pairwise}
In this appendix, we also show that constant $4$ contributes to the PGM error pairwise reduction. We prove a slightly more general $k$-wise lower bound for PGM pairwise reduction, whose case $k=2$ gives the desired sharpness.

\begin{thm}[(The PGM pairwise reduction)]
\label{thm:pgm-pairwise}
For every finite family $A_1,\ldots,A_N\ge 0$, the PGM error satisfies the dimension-independent pairwise bound
\bb
    \Err^{\pg}_N
    =
    4\sum_{i< j}\Err^{\pg}_{ij}.
\ee
\end{thm}
\begin{proof}
This result is a direct consequesnce of the pairwise reduction, i.e.,
\bb
    \Err^{\pg}_N \le 4\sum_{i<j} \Err^{\star}_{ij} \le 4\sum_{i<j}\Err^{\pg}_{ij},
\ee
with a trivial inequality $\Err^{\star}_{ij} \le \Err^{\pg}_{ij}$.
\end{proof}

In addition, in the following, we argue that the constant $4$ appears in the PGM pairwise reduction is also sharp. In fact, we prove a slightly more general $k$-wise lower bound, taking the case $k=2$ gives the desired result. 

For a subset $I\subseteq[N]\coloneqq \{1,\ldots,N\}$, write
\(
    \Err^{\pg}_I\coloneqq \Err^{\pg}\big(\{ A_i \}_{i\in I}\big).
\)
For $2\leq k\leq N$, define the normalized $k$-wise PGM error by
\bb
    \overline{\Err}^{\pg}_{N,k}
    \coloneqq 
    \binom{N-2}{k-2}^{-1}
    \sum_{\substack{I\subseteq[N]\\ |I|=k}}
    \Err^{\pg}_I,
    \label{eq:normalized-kwise-pgm-error}
\ee
which gives
\bb
    \overline{\Err}^{\pg}_{N,2}
    =
    \sum_{i< j}\Err^{\pg}(A_i,A_j)
    =
    2\sum_{i< j} \Tr A_i(A_i+A_j)^{-1/2} A_j(A_i+A_j)^{-1/2}.
\ee

\begin{prop}[(PGM error of Regular-simplex)]
\label{app:regular-simplex}
Let $2\leq k\leq N$. Suppose that a constant $C_{N,k}$ satisfies
\bb
    C_{N,k}=
    \inf\left\{
        C \ge 0: \Err^{\pg}_N \le C\,\overline{\Err}^{\pg}_{N,k} ~~\forall A_1,\ldots A_N \ge 0 
    \right\}
\ee
for every collection of positive semidefinite operators
$A_1,\ldots,A_N$. Then
\bb
    C_{N,k}
    \ge\left( 1+\sqrt{1-\frac{k}{N}} \right)^2.
\ee
\end{prop}
As a consequence, in particular, for every fixed $k$,
\bb
    \liminf_{N\to\infty}C_{N,k} = 4.
\ee

\begin{proof}[Proof of Proposition~
\ref{app:regular-simplex}]
Let $\{\ket{v_i}\}_{i=1}^N\subset\mathbb C^{N-1}$ be unit vectors forming a regular simplex, i.e., $\braket{v_i|v_j} \coloneqq  \alpha = -1/(N-1)$ with $i\neq j$. We define $A_i \coloneqq \ket{v_i}\bra{v_i}/N$.
For $I\subseteq[N]$, let $G_I$ be a Gram matrix, i.e.,
\bb
    (G_I)_{ij}
    =
    \begin{cases}
        1/N, & i=j,\\
        \alpha/N, & i\neq j.
    \end{cases}
\ee
By the block Gram representation used in the proof of
Theorem~\ref{thm:sharp-pairwise}, we have
\bb
    \Err^{\pg}_I
    =
    \sum_{\substack{ i\neq j\in I}}
    \left|
        \bigl(G_I^{1/2}\bigr)_{ij}
    \right|^2.
    \label{eq:err_I^pg}
\ee
It therefore suffices to compute the off-diagonal entries of
$G_I^{1/2}$. Now, fix $I\subseteq[N]$ with $|I|=k$. Given the symmetry of regular simplex, $\Err^{\pg}_I$ is the same for every such subset, and its Gram matrix reads
\bb
    G_I
    =
    \frac 1N
    \Big(
        (1-\alpha)\id_k+\alpha \mathds J_k 
    \Big),
\ee
where $\mathds J_k$ denotes the $k\times k$ all-ones matrix, which has only one non-zero eigenvalue $k$ with corresponding eigenspace projector $\hat J_k$. Therefore, we get
\bb
    G_I = 
    \frac 1N \Big( (1-\alpha + \alpha k)\, \hat{J}_k + (1-\alpha)\, \hat{J}_k^\perp \Big),
    \qquad
    G_I^{1/2} =
    \sqrt{\frac{N-k}{N(N-1)}} \, \hat J_k
    + \sqrt{\frac{1}{N-1}}\,\hat{J}_k^\perp.
\ee
Every off-diagonal element of $G_I^{1/2}$ is equal to
\bb
    \big(G_I^{1/2}\big)_{i\neq j} 
    &= 
    \frac 1k \left( \sqrt{\frac{N-k}{N(N-1)}} -\sqrt{\frac{1}{N-1}}\right)\\
    &=
    -\frac{1}{k\sqrt{N-1}} \left( 1-\sqrt{1-\frac{k}{N}} \right).
    \label{eq:steps-ref}
\ee
Given that all subsets of cardinality $k$ give the same error, by Eqs.~\eqref{eq:err_I^pg}, we get
\bb
    \Err^{\pg}_I
    &=
    k(k-1)\big(G_I^{1/2}\big)_{i\neq j}^2
    \\
    &=
    \frac{k-1}{k(N-1)}
    \left(
        1-\sqrt{1-\frac{k}{N}}
    \right)^2.
    \label{eq:ij-simplex-pgm-error}
\ee

Since all subsets of cardinality $k$ give the same error, definition in Eq.~\eqref{eq:normalized-kwise-pgm-error} gives
\bb
    \overline{\Err}^{\pg}_{N,k}
    &=
    \binom{N-2}{k-2}^{-1}\binom Nk
    \Err^{\pg}_I
    \\
    &=
    \frac{N(N-1)}{k(k-1)}
    \frac{k-1}{k(N-1)}
    \left(
        1-\sqrt{1-\frac{k}{N}}
    \right)^2
    \\
    &=
    \frac{N}{k^2}
    \left(
        1-\sqrt{1-\frac{k}{N}}
    \right)^2.
    \label{eq:simplex-kwise}
\ee

Analogously, for all the $N$ ensembles, following the same steps form Eqs.~\eqref{eq:err_I^pg} to \eqref{eq:steps-ref} yields
\bb
    \Err^{\pg}_N
    =
    N(N-1)
    \frac{1}{N^2(N-1)}
    =
    \frac 1N.
    \label{eq:full-simplex-pgm-error}
\ee

Combining Eqs.~\eqref{eq:simplex-kwise} and~\eqref{eq:full-simplex-pgm-error}, we get
\bb
    \frac{\Err^{\pg}_N}{ \overline{\Err}^{\pg}_{N,k} }
    &=
    \frac 1N \frac{k^2}{N} \left( 1-\sqrt{1-\frac{k}{N}} \right)^{-2} \\
    &=
    \left( 1+\sqrt{1-\frac{k}{N}} \right)^2.
\ee
Therefore, every constant $C_{N,k}$ satisfying the asserted $k$-wise inequality must obey
\bb
    C_{N,k}
    \ge
    \left(1+\sqrt{1-\frac{k}{N}}\right)^2.
    \label{eq:lower-bound-of-C}
\ee
For fixed $k$, the right-hand side converges to $4$ as $N\to\infty$, which proves the claim.
\end{proof}

\begin{rem}[(Conjecture of a possible $k$-wise hierarchy)]
The example of regular simplex also suggests a possible hierarchy of $k$-wise PGM reductions. 
and let $C_{N,k}$ denote the smallest constant such that
Theorem~\ref{thm:pgm-pairwise} gives $C_{N,2}\leq 4$, whereas $C_{N,N}=1$ holds trivially. The regular-simplex example suggests, for a fixed $N$, the monotonicity:
\bb
    4 \ge C_{N,2}\ge C_{N,3} \ge  \cdots  \ge  C_{N,N}=1,
\ee
where the proof of this hierarchy remain open.
\end{rem}

\end{widetext}

\end{document}